\documentclass{article}
\usepackage{cite}
\usepackage{amsmath,amssymb,amsfonts,amsthm}
\usepackage{algorithmic}
\usepackage{graphicx}
\usepackage{textcomp}
\usepackage{xcolor}
\def\BibTeX{{\rm B\kern-.05em{\sc i\kern-.025em b}\kern-.08em
    T\kern-.1667em\lower.7ex\hbox{E}\kern-.125emX}}

\theoremstyle{plain}
\newtheorem{theorem}{Theorem}
\newtheorem{lemma}[theorem]{Lemma}

\theoremstyle{definition}
\newtheorem{definition}[theorem]{Definition}

\newcommand{\Prop}{\mathsf{Prop}}
\newcommand{\card}{\mathsf{card}}
\newcommand{\until}{\mathbin{\mathsf{U}}}
\newcommand{\next}{\mathord\mathsf{X}}
\newcommand{\more}{\mathbin{\ll}}
\newcommand{\same}{\mathbin{\sim}}
\newcommand{\ltlp}{\mathsf{LTL}^{\more}}
\newcommand{\ltl}{\mathsf{LTL}}
\newcommand{\qin}{q_\text{in}}
\newcommand{\pre}{\mathsf{pre}}

\begin{document}

\title{The logic of temporal domination
}

\author{Thomas Studer}
\date{}

\maketitle

\begin{abstract}
In this short note, we are concerned with the fairness condition ``A~and~B hold almost equally often", which is important for specifying and verifying the correctness of non-terminating processes and protocols.
We introduce the logic of temporal domination, in which the above condition can be expressed. We present syntax and semantics of our logic and show that it is a proper extension of linear time temporal logic. In order to obtain this result, we rely on the corresponding result for k-counting automata.\\

\end{abstract}

\section{Introduction}

Temporal logic is a highly useful formalism for specifying and verifying correctness of computer programs, in particular for reasoning about non-terminating concurrent programs such as operating systems and communication protocols~\cite{pnueli,Emerson:1991}.
In this context, it is important to express how often a proposition~$A$ holds (an event~$A$ occurs) in relation to how often a proposition~$B$ holds (an event~$B$ occurs). For instance, the statement
\begin{equation}\label{eq:domination:1}
\text{$A$ and $B$ hold almost equally often}
\end{equation}
is a \emph{fairness} condition that guarantees that neither event $A$ nor~$B$ dominates the other event. 

In the present note, we introduce an extension $\ltlp$ of linear time temporal logic~$\ltl$ by a new binary temporal operator~$\more$. 
A formula $A \more B$ of $\ltlp$ roughly means that 
\[
\text{$B$ holds infinitely more often than $A$.} 
\]
Hence if we define 
\[
A \same B := \lnot (A \more B) \land \lnot (B \more A)
\rlap{\enspace,}
\]
then $A \same B$ expresses that~\eqref{eq:domination:1} holds.

We present syntax and semantics of $\ltlp$, study some basic properties, and prove that it is a proper extension of $\ltl$. This result relies on the close relationship of $\ltlp$ and $k$-counting automata.

Allred and Ultes-Nitsche~\cite{AllredU12,Ultes-NitscheA11} introduced
$k$-counting automata  
as recognizers for $\omega$-languages, i.e., languages over infinite words. They showed that the class of $\omega$-languages that are accepted by $k$-counting automata is a proper superclass of the $\omega$-regular languages~\cite{thomas90}.

The class of  $\omega$-regular languages is well studied. It can be defined in terms of $\omega$-regular expressions~\cite{5,thomas90}, B\"uchi automata~\cite{5}, Muller automata~\cite{15}, and many more formalisms. It is of high practical importance that $\omega$-regular languages are closed under Boolean operations, i.e., given $\omega$-regular languages, their union, intersection, and complement are $\omega$-regular, too (see~\cite{AllredU18} for a recent algorithm to compute the complement). Moreover, $\omega$-regular languages can be effectively tested for emptiness, i.e., given a description of an $\omega$-regular language in one of the formalisms above, it can be tested whether it is the empty language. As a consequence of the Boolean closure and the effective emptiness test, containment of $\omega$-regular languages can also be effectively tested by making use of
\[
A \subseteq B \quad\text{iff}\quad A \cap \overline{B} = \emptyset
\]
where $\overline{B}$ denotes the complement of $B$.

$k$-counting automaton refers to a deterministic finite-state machine model that is equipped with $k$-many counters. The acceptance condition makes statements about the boundedness or unboundedness of the counters in infinite runs of the automaton (see Definition~\ref{d:acc:1}). 
There is an important difference between $k$-counting automata and multi-counter machines~\cite{14,7,9}. The transition function in multi-counter machines depends on whether the counters are zero or not. This zero-test makes multi-counter machines Turing-complete and thus any interesting property about languages accepted by multi-counter machines is undecidable. The transition function in $k$-counting automata, in contrast, is independent of the values of the counters. Although it can change the counters, it does not read them. This makes  $k$-counting automata less powerful, but in exchange many interesting properties can remain decidable~\cite{AllredU12}.

Note that one of the properties that, unfortunately, cannot be expressed with $k$-counting automata is
\[
\text{there are never more $b$ than $a$ in any prefix of an $\omega$-word.}
\]
Such a property would make it possible, for example, to express that in an operating system one cannot kill more processes than one has started before. To express a property of this kind, we would need transitions that are only enabled when some counter is not zero. Because of closure under Boolean operations, this would also yield transitions that depend on a successful zero-test, which would be too strong (since it would give us undecidable multi-counter machines).

Allred and Ultes-Nitsche~\cite{AllredU12,Ultes-NitscheA11} show that the class of $\omega$-languages accepted by $k$-counting automata  
is closed under Boolean operations. Further they conjecture that the emptiness problem is decidable and they provide some formal arguments for this conjecture.

\section{$k$-counting automata}

The logic $\ltlp$ is inspired by, and closely related to, $k$-counting automata. In this section we will briefly recall the definition of $k$-counting automata and state their main properties.

\begin{definition}[$k$-counting automaton]
A \emph{$k$-counting automaton} is a tuple $(Q,C,\Sigma,\delta,\qin,\Phi)$ where
\begin{enumerate}
\item $Q$ is a finite set of states;
\item $C = \{c_0,\ldots,c_{k-1}\}$ is a set of $k$-many counters;
\item $\Sigma$ is a finite set of symbols;
\item $\delta : Q \times \Sigma \to Q \times 2^C \times 2^C$ is a transition function 
	such that for all $p \in Q$ and $a \in \Sigma$, 
\[
\delta(p,a) = (q,C_+,C_-) \text{ implies } C_+ \cap C_- = \emptyset;
\]
\item $\qin \in Q$ is an initial state; 
\item the acceptance condition $\Phi$ is a Boolean combination of the set of atomic propositions 
\[
\{ c_+ \ |\ c \in C\}\cup \{c _- \ |\ c \in C\}.
\]
\end{enumerate}
\end{definition}

\begin{definition}[Run]
Let $A=(Q,C,\Sigma,\delta,\qin,\Phi)$ be a $k$-counting automaton.
A \emph{counter valuation} $v$ is a mapping $v : C \to \mathbb{Z}$, where
$v(c) \in \mathbb{Z}$ denotes the value of counter~$c$ under the counter valuation $v$.
Let $V$ be the set of all counter valuations.

A \emph{run} $r \in (Q \times V)^\omega$ of $A$ on some $\omega$-word 
\[
x=x_0 x_1 x_2 \ldots \in \Sigma^\omega
\]
is a sequence of pairs 
\[
(q_0,v_0) (q_1,v_1) (q_2,v_2)\ldots
\]
such that
	\begin{gather*}
	q_0 = \qin \\
	(\forall  0\leq i\leq k-1)v_0(c_i)=0 \rlap{\enspace,}
	\end{gather*}
and, for all $j\geq 0$
\[
\delta(q_j,x_j) = (q_{j+1},C_+,C_-)
\]
and
\begin{gather*}
(\forall i \in C_+) v_{j+1}(c_i) = v_j(c_i)+1\rlap{\enspace,}\\
(\forall i \in C_-) v_{j+1}(c_i) = v_j(c_i)- 1\rlap{\enspace,}\\
(\forall i \in C \setminus (C_+ \cup C_-)) v_{j+1}(c_i) = v_j(c_i)
\rlap{\enspace.}
\end{gather*}
\end{definition}

\begin{definition}[Accepting run]\label{d:acc:1}
For the set 
\[
\{ c_+ \ |\ c \in C\}\cup \{c _- \ |\ c \in C\}
\]
of atomic propositions, we define satisfaction of atomic propositions in a run
\[
r = (q_0,v_0) (q_1,v_1) (q_2,v_2)\ldots
\]
as follows.
For all $c \in C$
\[
r \models c_+ \quad\text{if{f}}\quad
(\forall m \in \mathbb{Z})(\exists j>0)v_j(c)>m
\]
and
\[
r \models c_- \quad\text{if{f}}\quad
(\forall m \in \mathbb{Z})(\exists j>0)v_j(c)<m
\rlap{\enspace.}
\]
Based on this interpretation of atomic propositions and using the standard  semantics for Boolean connectives, we say a run~$r$ is \emph{accepting} if and only if 
\[
r\models \Phi
\rlap{\enspace.}
\]
\end{definition}
The above definition implies
\[
r \models \lnot c_+ \quad\text{if{f}}\quad
(\exists m \in \mathbb{Z})(\forall j>0)v_j(c)\leq m
\]
and
\[
r \models \lnot c_- \quad\text{if{f}}\quad
(\exists m \in \mathbb{Z})(\forall j>0)v_j(c)\geq m
\rlap{\enspace.}
\]
So an atomic proposition $c_+$ is satisfied if the counter $c$ is positively unbounded and $c_-$ is satisfied if $c$ is negatively unbounded. Similary, $\lnot c_+$ is satisfied if the counter $c$ is positively bounded and $\lnot c_-$ is satisfied if $c$ is negatively bounded. 

\begin{definition}[Accepted language]
An $\omega$-word $x$ is accepted by $A$ if and only if there exists an accepting run of $A$ on the word~$x$.
The \emph{$\omega$-language accepted by $A$} is the set of all accepted $\omega$-words.
\end{definition}

Let us now briefly summarize the main results about $k$-counting automata, see \cite{AllredU12,Ultes-NitscheA11} for proofs.

\begin{lemma}
Any regular $\omega$-language is accepted by some $k$-counting automaton, for some $k>0$.
\end{lemma}

For $w \in \{a,b\}^\ast$, let $[w]_a$ be the number of occurrences of symbol $a$ in $w$, and similarly for $[w]_b$.
Further for any $\omega$-word $x\in \Sigma^\omega$, let 
\[
\pre(x) := \{w \in \Sigma^\ast \ |\ (\exists y \in \Sigma^\omega) wy = x\}
\]
be the set of all finitely long prefixes of $x$.

\begin{lemma}\label{l:main:1}
The language
\begin{multline*}
\mathcal{L}_\omega := \{x \in \{a,b\}^\omega \ | \\
 (\exists m >0)(\forall w \in \pre(x)) \, | [w]_a-[w]_b | < m  \}
\end{multline*}
is a non-regular $\omega$-language that is accepted by some 1-counting automaton.
\end{lemma}

The language $\mathcal{L}_\omega$ in the above lemma contains all $\omega$-words~$x$ for which there exists $m >0$ such that in all prefixes $w$ of $x$ the number of occurrences of $a$ and the number of occurrences of $b$ does not differ by more than $m$.

\begin{theorem}
The class of $\omega$-languages that is accepted by $k$-counting automata is closed under Boolean operations.
\end{theorem}

\section{The logic $\ltlp$}

\subsection{Language}

The language of $\ltlp$ is the usual language of linear time temporal logic extended with a new binary temporal operator~$\more$. We read $A \more B$ as \emph{$A$ is dominated by $B$}.

We start with a set~$\Prop$ of atomic  proposition and use $P$ (possibly with subscript) to denote elements of $\Prop$.

The set of \emph{formulas} of $\ltlp$ is now inductively defined by
\begin{enumerate}
\item each atomic proposition is a formula;
\item if $A$ is a formula, then so is $\lnot A$;
\item if $A$ and $B$ are formulas, then so is $A \land B$;
\item if $A$ is a formula, then so is $\next A$;
\item if $A$ and $B$ are formulas, then so is $A \until B$;
\item if $A$ and $B$ are formulas, then so is $A \more B$.
\end{enumerate}

As usual, we set $\top := P \lor \lnot P$ for some fixed atomic proposition $P$. Further we define 
\[
\Diamond B := \top \until B \quad\text{and}\quad \Box B := \lnot \Diamond \lnot B
\rlap{\enspace.}
\]

If we drop the binary operator $\more$, then we obtain the usual linear the temporal logic $\ltl$.

\subsection{Semantics}

A \emph{model}, denoted $\sigma$, is an $\omega$-word over the alphabet $2^\Prop$, that is an $\omega$-word of subsets of $\Prop$. The symbol at position~$i$ is denoted by $\sigma_i$ and $\sigma_0$ denotes the first symbol of~$\sigma$.

We now define the relation $\models$ between models, natural numbers, and formulas by:
\begin{enumerate}
\item 
$\sigma, i \models P$ \quad if{f} \quad$P \in \sigma_i$ \qquad for $P\in \Prop$;
\item 
$\sigma, i \models \lnot A$ \quad if{f} \quad not $\sigma, i \models A$;
\item 
$\sigma, i \models A \land B$\quad  if{f} \quad  $\sigma, i \models A$ and $\sigma, i \models B$;
\item 
$\sigma, i \models \next A$\quad  if{f} \quad $\sigma, i+1 \models A$;
\item
$\sigma, i \models A \until B$ \quad if{f}  
\[
\exists j \geq i \big(\sigma, j \models B \text{ and } \forall i \leq k < j(\sigma,k\models A)\big)
\rlap{\enspace;}
\]
\item
$\sigma, i \models A \more B$ \quad if{f} 
\[
\forall b \exists j \big(\card(A_\sigma^{i,j}) + b \leq \card(B_\sigma^{i,j}) \big)
\]
where 
\[
A_\sigma^{i,j} := \{ k \in \omega \ |\ i\leq k \leq j \text{ and } \sigma,k\models A\}
\]
and  $\card(X)$ denotes the cardinality of a set $X$. 
\end{enumerate}
Given the interpretation of $\until$ as \emph{until}-operator and
recalling the definitions of the unary operators $\Box$ and $\Diamond$, we see that the formula~$\Diamond A$ means \emph{eventually $A$} and $\Box A$ means \emph{always~$A$}.

Let us briefly discuss the interpretation of $A \more B$.
We find that $A \more B$ holds  in a state $i$ if for any bound $b$ there is an interval $[i,j]$ such that the difference of the number of states in $[i,j]$ in which $B$ holds and the number of states in~$[i,j]$ in which $A$ holds is greater than or equal to $b$. Roughly we could say that $A \more B$ means that $B$ holds infinitely more often than~$A$.

Interesting is the negation of $A \more B$. We find that 
\[
\sigma, i \models \lnot (A \more B)
\]
if{f}
\[
\exists b \forall j \big(\card(A_\sigma^{i,j}) + b > \card(B_\sigma^{i,j}) 
\big)
\rlap{\enspace.}
\]
That means $\lnot (A \more B)$ is the case if there exists a bound $b$ such that the difference of the number of states in which $B$ holds and the number of states in which $A$ holds never exceeds $b$.

We can define a binary temporal operator $\same$ by
\[
A \same B := \lnot (A \more B) \land \lnot (B \more A)
\rlap{\enspace.}
\]
Given the above interpretation of $\lnot (A \more B)$, we can think of $A \same B$ as 
\[
\text{$A$ and $B$ hold almost equally often}.
\]
This reading is supported by the observation that $\top \same A$ means \emph{eventually always~$A$}. Indeed, we have the following lemma.

\begin{lemma}
The formula $\top \same A$ is  equivalent to the formula~$\Diamond\Box A$.
\end{lemma}
\begin{proof}
First we observe that $\top \more A$ is not satisfiable.
Therefore, and because of 
\[
\top \same A = \lnot (\top \more A) \land \lnot (A \more \top)
\rlap{\enspace,}
\]
we obtain that 
\[
\text{$\top \same A$ is  equivalent to $\lnot (A \more \top)$.}
\]
Hence it is enough to show that $\lnot (A \more \top)$ is equivalent to $\Diamond \Box A$.
We find that
$
\sigma, i \models \lnot (A \more \top)
$
if and only if
\[
\exists b>0 \forall j \big(\card(A_\sigma^{i,j}) + b > \card(\top_\sigma^{i,j}) 
\big)
\]
if and only if
\[
\exists b>0 \forall j \big(\card(A_\sigma^{i,j})  > (j+1-i)-b
\big)
\]
if and only if (choose $b' = i+b-1$)
\[
\exists b' \geq i \forall j \big(\card(A_\sigma^{i,j})  > j-b'
\big)
\]
if and only if (for the if-direction observe that $j<b'$ implies $\card(A_\sigma^{b',j}) = 0 \geq j-b'$)
\[
\exists b' \geq i \forall j'\geq b'  \big(\card(A_\sigma^{b',j'})  > j'-b'
\big)
\]
if and only if 
\[
\exists b' \geq i \forall j'\geq b' (\sigma, j' \models A)
\]
if and only if $\sigma, i \models \Diamond \Box A$.
\end{proof}

We now immediately see that
$A \more \top$ is equivalent to $\Box\Diamond \lnot A$, which relates to the $\forall\exists$ quantification pattern in the interpretation of the $\more$-operator.

\section{Beyond $\ltl$}

In this section, we show that $\ltlp$ is strictly more expressive than $\ltl$.
First, we define 
the language~$\mathcal{L}_A$ of a formula~$A$ to consists of all $\omega$-words that are models of $A$. Formally we define this as follows.

\begin{definition}
Let $A$ be a formula with atomic propositions $P_1,\ldots,P_n$. 
The \emph{language~$\mathcal{L}_A$ of $A$} is the $\omega$-language over the alphabet $\Sigma := 2^{\{P_1,\ldots,P_n\}}$ given by
\[
\mathcal{L}_A := \{ \sigma \in \Sigma^\omega \ |\ \sigma,0 \models A\}
\rlap{\enspace.}
\]
\end{definition}

Now we establish that there is an $\ltlp$-formula $A$ such that there is no $\ltl$-formula $B$ that has the same language as~$A$.
Thus  $\ltlp$ is strictly more expressive than $\ltl$.

\begin{theorem}
There is a formula $A$ of\/ $\ltlp$ such that for all formulas $B$ of\/ $\ltl$ we have
\[
\mathcal{L}_A \neq \mathcal{L}_B
\rlap{\enspace.}
\]
\end{theorem}
\begin{proof}
Let $A$ be the $\ltlp$-formula $\lnot ((\lnot P) \more P)$. Then $\mathcal{L}_A$ is a language over the alphabet $\Sigma := 2^{\{P\}}$. We let $a:=\{P\}$ and $b:=\emptyset$. We find that
\[
\sigma \in \mathcal{L}_A 
\quad\text{if{f}}\quad
\exists b \forall j \big(\card(A_\sigma^{0,j}) + b > \card(B_\sigma^{0,j}) 
\big)
\rlap{\enspace.}
\]
Hence $\mathcal{L}_A  = \mathcal{L}_\omega$ where $\mathcal{L}_\omega$ is as in Lemma~\ref{l:main:1}.
Therefore, we have that 
\begin{equation}\label{eq:th:1}
\text{
$\mathcal{L}_A$ is a non-regular $\omega$-language.
}
\end{equation}
For every $\ltl$-formula $B$ we know that 
\begin{equation}\label{eq:th:2}
\text{
$\mathcal{L}_B$ is a star-free regular $\omega$-language, 
}
\end{equation}
see for instance~\cite{thomas90,Ladner1977281}.

By \eqref{eq:th:1} and \eqref{eq:th:2} we conclude that for each $\ltl$-formula $B$ we have $\mathcal{L}_A \neq \mathcal{L}_B$.
\end{proof}

\section{Conclusion}

We  have introduced the logic of temporal domination $\ltlp$. After defining its language and semantics, we have established that  $\ltlp$ is a proper extension of linear time temporal logic $\ltl$. This results makes essential use of the close relationship between the logic of temporal domination~$\ltlp$ and $k$-counting automata, which recognize a proper extension of regular $\omega$-languages.

Of course, this small note is only a first step in exploring the logic of temporal domination. There are several interesting open question that will be addressed in future work.
 
\begin{enumerate}
\item  
Develop a (finitary or infinitary) deductive system for $\ltlp$.
\item Establish soundness and completeness for that deductive system.
\item Show that $\ltlp$ is decidable. This is probably  closely related to the problem of showing that the emptiness test for $k$-counting automata is computable.
\item Study the model checking problem for $\ltlp$. What is its complexity? Can it be implemented efficiently?
\item Examine the exact relationship between $\ltlp$ and $k$-counting automata. Do they have the same expressive power?
\end{enumerate}

\section*{Acknowledgements}
We thank the anonymous referees for many helpful comments.
This work is supported by the Swiss National Science Foundation grant 200021\_165549.

\end{document}